\documentclass{lmcs}
\pdfoutput=1

\usepackage{lastpage}
\renewcommand{\dOi}{14(4:14)2018}
\lmcsheading{}{1--\pageref{LastPage}}{}{}%
{Jul.~22,~2015}{Nov.~15,~2018}{}

\usepackage[T1]{fontenc}
\usepackage{amsmath,amsfonts,amssymb,stmaryrd}
\usepackage{xspace}
\usepackage{cmll}
\usepackage{mathtools}
\usepackage{listings}
\usepackage{stmaryrd}
\usepackage{mathpartir}
\usepackage{eurosym}
\usepackage{centernot}

\usepackage{tikz}
\usetikzlibrary{calc}
\usetikzlibrary{arrows}
\usetikzlibrary{positioning}
\usetikzlibrary{fadings}
\usetikzlibrary{chains}
\usetikzlibrary{fit}
\usetikzlibrary{intersections}

\newcommand{\tyder}{\ensuremath{\mathcal{D}}}
\newcommand{\tyderT}[2]{\ensuremath{\mathcal{D}_{#1}::#2}}
\newcommand{\tyderHs}[2]{\ensuremath{\mathcal{D}_{#1}[\cdot_{#2}]}}
\newcommand{\tyderH}[3]{\ensuremath{\mathcal{D}_{#1}[\cdot_{#2}]::#3}}
\newcommand{\tyderCs}[2]{\ensuremath{\mathcal{D}_{#1}[#2]}}
\newcommand{\tyderC}[3]{\ensuremath{\mathcal{D}_{#1}[#2]::#3}}
\theoremstyle{plain}

\newcommand{\Case}{\noindent\textbf{Case}\ }
\newcommand{\ie}{i.e.} 
\newcommand{\eg}{e.g.}

\newcommand{\etal}{et al.\xspace}
\newcommand{\qt}[1]{``#1''}
\newcommand{\pic}{\ensuremath{\pi\textrm{-calculus}}}

\newcommand{\tsf}[1]{\ensuremath{\mathsf{#1}}}
\newcommand{\txsf}[1]{\ensuremath{\textsf{#1}}}

\newcommand{\grmeq}{\ensuremath{::=}}
\newcommand{\aconv}{\ensuremath{\alpha}-conversion}
\newcommand{\ovl}[1]{\ensuremath{\overline{#1}}}

\newcommand{\blk}[1]{\ensuremath{(#1)}} 
\newcommand{\aact}{\ensuremath{\alpha}}

\newcommand{\setsep}{:} 

\newcommand{\set}[1]{\ensuremath{\left\{#1\right\}}}

\newcommand{\kw}[1]{\ensuremath{\mathsf{#1}}}        
\def\R{\longrightarrow} 												  
\def\notR{\centernot\R}
\newcommand{\defeq}{\ensuremath{\,\doteq\,}}

\newcommand{\dual}[1]{\ensuremath{\overline{#1}}}        
\newcommand{\tp}{\ensuremath{\colon}}  						  

\def\recv#1{\wn{#1}}
\def\send#1{\oc {#1}}
\newcommand{\bra}{{\with}}
\newcommand{\sel}{{\oplus}}
\newcommand{\setT}[3]{\set{ l_{#2} \tp {#1} }_{#2 \in #3}}     

\newcommand{\sessend}{\ensuremath{\mathsf{end}}}
\newcommand{\aT}[1]{\ensuremath{\mathsf{acc} \; #1}}		
\newcommand{\rT}[1]{\ensuremath{\mathsf{req} \; #1}}		
\newcommand{\hast}{\ensuremath{\vdash}}

\newcommand{\seq}[2]{\ensuremath{#2 \hast #1}}
\newcommand{\trule}[3]{\ensuremath{ \inferrule[{\kw{(#1)}}]{#2}{#3} }}

\newcommand{\norequests}{\mathsf{norequests}}

\newcommand{\new}{\ensuremath{\boldsymbol \nu}}
\newcommand{\newc}[1]{\ensuremath{(\new #1)}}

\newcommand{\kils}[1]{\ensuremath{#1 \lightning}}
\newcommand{\nil}{\ensuremath{\mathbf{0}}}
\newcommand{\prl}{\ensuremath{\mid}}                          
\newcommand{\subs}[2]{\ensuremath{ \{ \raisebox{2.6pt}{{\small $ #1 $}} \!\, / \!\, \raisebox{-1.0pt}{{\small $#2$}} \} }} 
\newcommand{\spref}{\ensuremath{\rho}}
\newcommand{\sout}[2]{\ensuremath{\ovl{#1} #2}}
\newcommand{\sinp}[2]{\ensuremath{#1(#2)}}
\newcommand{\req}[2]{\ensuremath{ \tsf{req}\: \ovl{#1} {#2} }}
\newcommand{\reqi}[3]{\ensuremath{ \tsf{req}\: \ovl{#1}_{#2} {#3} }}
\newcommand{\acc}[2]{\ensuremath{ \tsf{acc}\: \sinp{#1}{#2} }}
\newcommand{\choiceI}[2]{\ensuremath{ \ovl #1 \triangleleft {l_#2} }}
\newcommand{\choice}[2]{\ensuremath{\ovl #1 \triangleleft \tsf{#2} }}
\newcommand{\branchI}[4]{\ensuremath{#1 \triangleright \set{ {l_{#3}} . #2_{#3} }_{#3 \in #4}}}
\newcommand{\branch}[2]{\ensuremath{#1 \triangleright \set{#2}}}

\newcommand{\orBra}{\ensuremath{ \talloblong }}

\newcommand{\methK}[2]{\ensuremath{ \tsf{#1}.#2 }}
\newcommand{\trycatch}[2]{\ensuremath{ \tsf{do} \; #1 \; \tsf{catch}
    \; #2 }}
\newcommand{\docatchk}{$\tsf{do}$-$\tsf{catch}$} 
\newcommand{\gctx}[1]{\ensuremath{H[\,#1\,]}}

\newcommand{\gctxi}[2]{\ensuremath{H_{#1}[\,#2\,]}}

\newcommand{\inact}[1]{\ensuremath{\mathsf{inactive}(#1)}}
\newcommand{\acti}[1]{\ensuremath{\mathsf{active}(#1)}}

\newcommand{\cata}[1]{\ensuremath{\llparenthesis\, #1 \,\rrparenthesis}}

\newcommand{\barb}[2]{\ensuremath{#2 {\downarrow_{#1}}}}
\newcommand{\wbarb}[2]{\ensuremath{#2 {\Downarrow_{#1}}}}

\newif\ifshownodes
\pgfkeys{/show nodes/.is if=shownodes}
\pgfkeys{/show nodes=true}
\pgfkeys{/show nodes=false}

\tikzfading[name=fade right,
left color=transparent!0,
right color=transparent!100]

\tikzfading[name=fade left,
right color=transparent!0,
left color=transparent!100]

\tikzfading[name=fade out,
inner color=transparent!50,
outer color=blue]

\tikzset{%
role/.style={line width=1.2pt, color=black, rounded corners=0.2cm, line cap=round} , 
comm/.style={role, color=black},
crossing comm/.style={comm, double distance=1.2pt, draw=white, double=black}, 
}%

\pgfdeclarelayer{background}
\pgfdeclarelayer{foreground}
\pgfsetlayers{background,main,foreground}

\newdimen\xA
\newdimen\xB
\newdimen\xMid
\newdimen\yA
\newdimen\yB
\newdimen\yLength

\newdimen\deltaY
\newdimen\delta

\newcommand{\aRole}[4]{%

\pgfmathsetlength{\deltaY}{0pt} 
\pgfmathsetlength{\delta}{0pt} 

\path (#2) node(#1-node) [role, draw, color=black, anchor=south] {{\vphantom{Ay}\texttt{#1}}} ;

\path (#2) coordinate (#1-top) ; 

\draw[role, dotted, color=black] (#1-top) -- ($ (#1-top) + (0,-#3) + (0,-0.2cm) $) ;

\pgfmathsetlength\deltaY{#3/(#4)} ;

    \foreach \i [count=\j] in {1,...,#4}
           { 
			
			\pgfmathaddtolength\delta{\deltaY}
	         \global\delta=\delta 
           
            \path ($ (#1-top) + (0,-\delta) $) coordinate (#1-\i) ;
           };

\path ($ (#1-#4) + (0,-0.0cm) $) coordinate (#1-bot) ;

}%

\newcommand{\aMsg}[5][midway]{%
    \draw[role,->] (#2) -- (#3) node[above, #1] {$#4 : \mathsf{#5}$};
}

\newcommand{\aCrossingMsg}[7][midway]{%
   \path ($ (#2) +(0.2cm,0) $) coordinate (#2-#3-L); 
   \path ($ (#3) +(-0.2cm,0) $) coordinate (#2-#3-R);
   \begin{pgfonlayer}{foreground}
    \draw[role,#6] (#2) -- (#2-#3-L) ;
    \draw[role,#7] (#2-#3-R) -- (#3) ;
    \end{pgfonlayer}
    \draw[role,crossing comm] (#2-#3-L) -- (#2-#3-R) node[above, #1] {$#4 : \mathsf{#5}$} ;
}

\begin{document}

\title{Affine Sessions}
\author[Mostrous]{Dimitris Mostrous}
\address{University of Lisbon, Faculty of Sciences and LASIGE Lisbon, Portugal}
\email{dimitris.mostrous@googlemail.com}
\author[Vasconcelos]{Vasco Thudichum Vasconcelos}
\address{\vskip-7pt}
\email{vmvasconcelos@ciencias.ulisboa.pt}
\titlecomment{An extended abstract of this paper appeared in COORDINATION 2014.}

\begin{abstract}
  Session types describe the structure of communications implemented by
channels.  In particular, they prescribe the sequence of
communications, whether they are input or output actions, and the type
of value exchanged.  Crucial to any language with session types is the
notion of linearity, which is essential to ensure that channels
exhibit the behaviour prescribed by their type without interference in
the presence of concurrency.  In this work we relax the condition of
linearity to that of affinity, by which channels exhibit at most the
behaviour prescribed by their types.  This more liberal setting allows
us to incorporate an elegant error handling mechanism which simplifies
and improves related works on exceptions.  Moreover, our treatment
does not affect the progress properties of the language: sessions
never get stuck.


\end{abstract}

\maketitle 

\section{Introduction}\label{sec:intro}

A session is \emph{a semantically atomic chain of communication
  actions which can interleave with other such chains freely, for
  high-level abstraction of interaction-based computing}~\cite{THK}.
\emph{Session types}~\cite{honda.vasconcelos.kubo:language-primitives}
capture this intuition as a description of the structure of a
protocol, in the simplest case between two programs (binary sessions).
This description consists of types that indicate whether a
communication channel will next perform an output or input action, the
type of the value to send or receive, and what to do next,
inductively.  

For example,
$\send {\tsf{nat}}.\send {\tsf{string}}.\recv {\tsf{bool}}.\sessend$
is the type of a channel that will first send a value of type
\tsf{nat}, then one of type \tsf{string}, then receive a value of
type \tsf{bool}, and nothing more.  This type can be materialised by
the \pic~\cite{MilnerR:calmp1} process
$P_1 \defeq \sout a 5 . \sout {a} {\,\text{\qt{hello}}} . \sinp a x . \nil$.  
The dual of the previous type is
$\recv {\tsf{nat}}.\recv {\tsf{string}}.\send {\tsf{bool}}.\sessend$,
and can be implemented by
$P_2 \defeq \sinp b x . \sinp b y . \sout b {(x + 1 < 2)} . \nil$.
To compose two processes and enable them to communicate, 
we use a \emph{double binder}~\cite{Vasconcelos2012}. 
For the above example, we can write 
$\newc{ab}\blk{ P_1 \prl P_2 }$, indicating that $a$ and $b$ are the two 
\emph{endpoints} of the same channel. 
The double binder guides reduction, so that we have 
$\newc{ab}\blk{ P_1 \prl P_2 } \R 
\newc{ab}\blk{\sout {a} {\,\text{\qt{hello}}} . \sinp a x . \nil} \prl \sinp b y . \sout b {(5 + 1 < 2)} . \nil$. 
In a well-typed term, the endpoints of a channel must have 
complementary (or \emph{dual}) types, so that an input on one will
match an output on the other, and vice versa.   
This is the case for $a$ and $b$, above. 

Beyond the basic input/output types, sessions typically provide
constructors for alternative sub-protocols, which are very useful for
structured
interaction. 
For example, type
$\bra \set{ \tsf{go} \tp T_1, \tsf{cancel} \tp T_2 }$ can be assigned
to an (external) choice
$\branch a { \methK {go} {Q_1} \orBra \methK {cancel} {Q_2} }$, a
process that offers the choice \tsf{go} and then $Q_1$ or \tsf{cancel}
and then $Q_2$.  The dual type, where $\dual T$ denotes $T$ with an
alternation of all constructors, is
$\sel \set{ \tsf{go} \tp \dual{T_1}, \tsf{cancel} \tp \dual{T_2} }$,
and corresponds to a process that will make a (internal) choice,
either $\choice b {go}.R_1$ or $\choice b {cancel}.R_2$.  In the first
case the two processes will continue as $Q_1$ and $R_1$, respectively.

\subsection*{From Linearity to Affinity}

To ensure that sequenced interactions take place in the prescribed
order, session typing relies crucially on the notion of
\emph{linearity}~\cite{Girard87}.  However, instead of requiring each
endpoint to appear exactly once in a term, which is the standard
notion of linearity, session systems only require that an endpoint can
interact once at any given moment.  Both channel ends $a$ and $b$ in
processes $P_1$ and $P_2$ are linear in this sense.  To see why this
condition is required, imagine that we write the first process as
$P_1' \defeq \sout a 5 . \nil \prl \sout {a} {\,\text{\qt{hello}}}
. \sinp a x . \nil$.  Now, $a$ does not appear linearly in $P_1'$
since there are two possible outputs ready to fire. The net effect is
that $P_2$ can receive a \qt{hello} first, which would clearly be
unsound and would most likely raise an error in any programming
environment.
We only relax this condition in one case: two outputs (of the same type) are allowed 
in parallel when the dual endpoint is a replicated input. 

It is because of linearity, as explained above, that sessions can be used to 
structure protocols with sequences of inputs and outputs, without losing type safety.  
However, linearity is a rather rigid condition, because it demands that 
everything in the description of a session type \emph{must} be implemented 
by an endpoint with that type. 
%
In real world situations, interactions are structured but can be
aborted at any time.  For example, an online store should be prepared
for clients that get disconnected, that close their web browsers, or
for general \emph{errors} that abruptly severe the expected pattern of
interaction.

In this work we address the above issue. 
In technical terms, we relax the condition of linearity to that of
\emph{affinity}, so that endpoints can perform less interactions than
the ones prescribed by their session type.
However, a naive introduction of affinity can leave programs in a
stuck state: let us re-write $P_1$ 
into
$P_1'' \defeq \sout a 5 . \sout {a} {\,\text{\qt{hello}}} . \nil$,
\ie, without the final input $\sinp a x$; then, after two
communications process $\newc{ab}\blk{ P_1'' \prl P_2 }$ will be stuck
trying to perform the output $\sout b {(5 + 1 < 2)} . \nil$.
We want to be able to perform only an initial part of a session, but we also want to ensure that 
processes do not get stuck waiting for 
communications that will never take place. 
Our solution is to introduce a new kind of communication action written 
$\kils a$, which reads \emph{cancel} $a$. 
This action is used to explicitly signal that a 
session has finished, so that 
communications on the other endpoint 
can also be cancelled and computation can proceed. 
For example, we can replace $P_1$ with 
$P_1^c \defeq \sout a 5 . \sout {a} {\,\text{\qt{hello}}} . \kils a$, 
and after two steps $\newc{ab}\blk{ P_1^c \prl P_2 }$ becomes  
$\newc{ab}\blk{ \kils a \prl \sout b {(5 + 1 < 2)} . \nil}$, which 
reduces (modulo structural equivalence) to $\nil$. 
%

Our development is inspired by Affine Logic, the variation of Linear Logic with unrestricted 
weakening. 
The work by Asperti~\cite{Asperti:2002:ILAL}, which studies 
Proof Nets for Affine Logic, 
shows that weakening corresponds to an actual connective with specific behaviour. 
In particular, this connective performs the weakening 
step by step, progressing through the dependencies of a proof, 
and removing all that must be removed. 
This is exactly what $\kils a$ represents. 

%
We take the idea of affinity a step further: if cancellation of a
session is explicit, we can treat it as an \emph{exception}, and for
this we introduce a \docatchk{} construct that can provide an
alternative behaviour activated when a cancellation is encountered.
For example,
$\newc{ab}\blk{ \trycatch {\sout a {(5 + 1 < 2)} . \nil} {P} \prl
  \kils b}$ will result in the \emph{replacement} of
$\sout a {(5 + 1 < 2)} . \nil$ with the exception handler $P$.  Note
that a \docatchk{} is not the same as the try-catch commonly found in
sequential languages: it does not define a persistent scope that
captures exceptions from the inside, but rather it applies to the
first communication and is activated by exceptions from the outside
(as in the previous example).  Thus,
$\newc{ab}\blk{ \trycatch {\sout a {(5 + 1 < 2)} . \nil} {P} \prl
  \sinp b x.\nil}$ becomes $\nil$, because the communication was
successful.


The outline of the rest of the paper is as follows. The next section
presents affine sessions in action.  Section~\ref{sec:calculus}
introduces the calculus of affine sessions, Section~\ref{sec:typing}
its typing system, and Section~\ref{sec:properties} the main
properties. Section~\ref{sec:related-work} discusses related works and
future plans. The appendix contains the proof of the Subject Reduction
theorem.


\section{Affine Sessions by Example}
\label{sec:example}

\begin{figure}
\begin{center}
\vspace*{10pt}
{\small \begin{tikzpicture}

\coordinate (c) at (0,0);
\path (c) +(4cm,0) coordinate (c1);
\path (c1) +(4cm,0) coordinate (c2);

\aRole{Buyer}{c}{6cm}{9} ;
\aRole{Seller}{c1}{6cm}{9} ;
\aRole{Bank}{c2}{6cm}{9} ;

\aMsg{Buyer-1}{Seller-1}{b}{\text{\qt{Proofs and Types}}}
\aMsg{Seller-2}{Buyer-2}{b}{\text{\EUR{178}}} 
\aMsg{Buyer-3}{Seller-3}{b}{\tsf{select\ buy}}

\aMsg{Seller-4}{Bank-4}{c}{\text{\EUR{178}}}
\aMsg{Seller-5}{Bank-5}{c}{\textit b}

\aCrossingMsg[midway, fill=white]{Buyer-6}{Bank-6}{b}{\textit{ccard}}{}{->}

\aMsg{Bank-7}{Seller-7}{c}{\textit b}
\aMsg{Bank-8}{Seller-8}{c}{\tsf{select\ accepted}} 
\aMsg{Seller-9}{Buyer-9}{b}{\tsf{select\ accepted}} 

\end{tikzpicture} }
\end{center}
\caption{Sequence Diagram for Succesful Book Purchase}\label{fig:buy1}
\end{figure}


We describe a simple interaction comprising three
processes---\tsf{Buyer}, \tsf{Seller}, and \tsf{Bank}---that
implements a book purchase.  The buyer sends the title of a book,
receives the price, and chooses either to buy or to cancel.  If the
buyer decides to buy the book, the credit card information is sent
over the session, and the buyer is informed whether or not the
transaction was successful.  The diagram in Figure~\ref{fig:buy1}
shows the interactions of a specific purchase.

We now show how this scenario can be implemented using sessions, and
how our treatment of affinity can be used to enable a more concise and
natural handling of exceptional outcomes.  Our language is an almost
standard \pic\ where replication is written $\acc a x . P$ and plays
the role of \qt{accept} in session
terminology~\cite{honda.vasconcelos.kubo:language-primitives}.
Dually, an output that activates a replication is written
$\req a b . P$, and is called a \qt{request.} 
Channels are described by two distinct identifiers, denoting their two
endpoints and introduced by 
$\newc{ab}P$~\cite{Vasconcelos2012}.

We use some standard language constructs that can be easily encoded in
\pic, such as $\sout a e$ for the output of the value obtained by
evaluating the expression $e$, and
$\tsf{if} \: t \: \tsf{then} \: P \: \tsf{else} \: Q$ for a conditional
expression. The latter is an abbreviation of a new session
$\newc{ab}\blk{ \branch a {\methK {true} P \orBra \methK {false} Q}
  \prl R}$ where $R$ represents the test $t$ and evaluates to
$\choice b {true}$ or $\choice b {false}$.
An implementation of the interaction in Figure~\ref{fig:buy1} is: 
\begin{gather*}
\newc{seller_1 seller_2, bank_1 bank_2}\blk{\: \tsf{Buyer} \prl \tsf{Seller} \prl \tsf{Bank} \:}
\end{gather*}%
where: 
\begin{gather*}
\begin{array}{rcl}
\tsf{Buyer} & \defeq & \newc {bb'}  \left( 
     \begin{array}{l}
         \req {\textit{seller}_1} {\,b'} \prl \sout b {\,\text{\qt{Proofs and Types}}} . \sinp b {\textit{price}} . 
         \tsf{if} \: \textit{price} < 200 \\ 
         \tsf{then} \: \choice b {buy} . \sout b {\,\text{\it{ccard}}} . 
         \branch b { \methK {accepted} {P} \orBra \methK {rejected} {Q} }  \: 
         \tsf{else} \: \choice b {cancel}   
     \end{array} \right)
\\ \\
\tsf{Seller} & \defeq & \acc {\textit{seller}_2} {b} . \left( 
     \begin{array}{l}
            \sinp b {\textit{prod}} . \sout b {\,\text{\it{price}}(\textit{prod})} .   \\
                      b \triangleright \{ \: \tsf{buy} .  \newc {k k'} (  \req {\textit{bank}_1} {\,k'} \prl \sout k {\,\text{\it{price}}(\textit{prod})} . \sout k b . \sinp k {b'} . \\ 
                            \qquad k \triangleright \{ \tsf{accepted}.\choice {b'} {accepted} \orBra \tsf{rejected} .  \choice {b'} {rejected} \}  )  \\
                              \:\,\quad \orBra\: \methK {cancel} {\nil} \}  \\
     \end{array} 
     \right)
\\ \\
 \tsf{Bank} & \defeq & \acc {\textit{bank}_2} {k} . \left( 
     \begin{array}{l}
            \sinp k {\textit{amount}} . \sinp k {b} . \sinp b {\textit{card}} . \sout k b .  \\ 
            \tsf{if} \: \text{\it{charge}}(\textit{amount}, \textit{card}) \: \tsf{then} \:
                  \choice k {accepted} \: \tsf{else} \: \choice k {rejected} 
     \end{array} 
     \right)
\end{array}
\end{gather*}

First we note how sessions are established.  For example, in
\tsf{Buyer} fresh channel end $b'$ is sent to \tsf{Seller} via the
request $\req {\textit{seller}_1} {\,b'}$, while the other end, $b$,
is kept in the \tsf{Buyer} for further interaction.  The identifiers
$b$ and $b'$ are the two endpoints of a session, and it is easy to
check that the interactions match perfectly.
Another point is the \emph{borrowing} of the session $b$ from
\tsf{Seller} to \tsf{Bank}, with subprocess $\sout k b . \sinp k {b'}$
at the \tsf{Seller} process, and
$\sinp k b . \sinp b {\textit{card}} . \sout k b$ at the \tsf{Bank},
so that the credit card information is received directly by \tsf{Bank}; see also
Figure~\ref{fig:buy1}.

A more robust variation of \tsf{Seller} could utilise the \docatchk{}
mechanism to account for the possibility of the \tsf{Bank} not being
available. In this case, the seller would provide an alternative
payment provider.  Concretely, we can substitute
$\req {\textit{bank}_1} {\,k'}$ in \tsf{Seller} with
$\trycatch {\req {\textit{bank}_1} {\,k'}} {\req {\textit{paymate}}
  {\,k'}}$, so that a failure to use the bank service (triggered by
$\kils {\textit{bank}_2}$) will activate
$\req {\textit{paymate}} {\,k'}$ and the protocol has a chance to
complete successfully.

The \tsf{Buyer} might also benefit from our notion of exception
handling.  As an example we show an adaptation that catches a
cancellation at the last communication of the \tsf{buy} branch and
prints an informative message:
 \begin{gather*}
\tsf{BuyerMsg} \: \defeq \: \newc {bb'}  \left( 
     \begin{array}{l}
         \req {\textit{seller}_1} {\,b'} \prl \sout b {\,\text{\qt{Proofs and Types}}} . \sinp b {\textit{price}} . \\
         \tsf{if} \: \textit{price} < 200 \: \tsf{then} \: \choice b {buy} . \sout b {\,\text{\tsf{ccard}}} . \\
         \quad \tsf{do} \: 
               \branch b { \methK {accepted} {P} \orBra \methK {rejected} {Q} }  \\ 
          \quad \tsf{catch} \: {\req {\textit{print}} {\,\text{\qt{An error occurred}}}} \\
         \tsf{else} \: \choice b {cancel}   
     \end{array} \right)
\end{gather*}%

As mentioned in the Introduction, a \docatchk{} on a given
communication does not catch subsequent cancellations.  For instance,
if in the above example the \docatchk{} was placed around
$\sout b {\,\text{\qt{Proofs and Types}}}$, then any $\kils {b'}$
generated \emph{after} this output has been read would be uncaught,
since $\req {\textit{print}} {\,\text{\qt{An error occurred}}}$ would
have been already discarded.
However, a \docatchk{} does catch cancellations emitted \emph{before}
the point of definition, so it should be placed near the
end of a protocol if we just want a single exception handler that
catches everything.  In general, our mechanism is very fine-grained,
and a single session can have multiple, nested \docatchk{} on crucial
points of communication and with distinct alternative behaviours.
 
Note also that cancellation 
can be very useful in itself, even without the \docatchk{} mechanism. 
Here are two ways to implement a process that starts a protocol with \tsf{Seller} only 
to obtain the price of a book and use it in $R$: 
\begin{align*}
\tsf{CheckPriceA} \defeq & \newc {bb'}(
         \req {\textit{seller}_1} {\,b'} \prl \sout b
       {\,\text{\qt{Principia Mathematica}}} . \sinp b {\textit{price}} . ( \choice b {cancel} \prl R ) )
\\
\tsf{CheckPriceB} \defeq & \newc {bb'}(
         \req {\textit{seller}_1} {\,b'} \prl \sout b
       {\,\text{\qt{Introduction to Metamathematics}}} . \sinp b {\textit{price}} . ( \kils b \prl R )  )
\end{align*}
Both the above processes can be typed.  However, the first requires a
knowledge of the protocol, which in that case includes an exit point
(branch \tsf{cancel}), while the second is completely transparent.
For example, imagine a buyer that selects \tsf{buy} by accident and
then wishes to cancel the purchase: without cancellation this is
impossible because such behavior is not \emph{predicted} by the
session type; with cancellation it is extremely simple, as shown
below.
\begin{gather*}
\tsf{BuyerCancel} \: \defeq \: \newc {bb'}(
         \req {\textit{seller}_1} {\,b'} \prl \sout b {\,\text{\qt{Tractatus Logico-Philosophicus}}} . \sinp b {\textit{price}} .  \choice b {buy} . \kils b) 
\end{gather*}


\section{The Process Calculus of Affine Sessions}
\label{sec:calculus}

This section introduces our language, its syntax and operational
semantics.

\subsection*{Syntax} 

The language we work with, shown in Figure~\ref{fig:syntax}, is a
small extension of standard \pic~\cite{MilnerR:calmp1}. We rely on a
denumerable set of \emph{variables}, denoted by lower case roman
letters.
As for processes, instead of the standard restriction $\newc a P$, we
use double binders~\cite{Vasconcelos2012} in the form $\newc{ab}P$,
which are similar to \emph{polarities}~\cite{GaySJ:substp}, and enable
syntactically distinguishing the two endpoints of a session.
For technical convenience we shall consider all indexing sets $I$ to
be non-empty, finite, and totally ordered, so that we can speak, \eg,
of the maximum element.  Also for technical convenience, we separate
the prefixes denoted by $\spref$, \ie, all communication actions
except for accept (replication).  We only added two non-standard
constructs: the \emph{cancellation} $\kils a$ and the \emph{do-catch}
construct that captures a cancellation, written $\trycatch \spref P$.


Parentheses introduce the \emph{bindings} in the language: variable
$x$ is bound in processes $\sinp ax.P$ and $\acc ax.P$; both variables
$x$ and $y$ are bound in process $\newc {xy}P$. The notions of free
and bound variables as well as that of substitution (of $x$ by $a$ in $P$,
notation $P\subs ax$) are defined accordingly.
We follow Barendregt's variable convention, whereby all variables in
binding occurrences in any mathematical context are pairwise distinct
and distinct from the free variables.

\begin{figure}
\begin{center}
\begin{gather*}
\begin{array}{rclcr}
\spref & \grmeq & \sinp a {x} . P 				&\qquad& (\text{input})			\\[.5ex]
   & \prl      &  \sout a b . P          				    && (\text{output}) 		\\[.5ex]
   & \prl      &  \branchI a P i I       			    && (\text{branching}) 	\\[.5ex]
   & \prl      &  \choiceI a {k} . P    			    	&& (\text{selection})		\\[.5ex]
    & \prl      &  \req a {b} . P      				&& (\text{request}) 		\\[.5ex]
   \\  
P   &  \grmeq      &  \spref			&& (\text{prefix}) 		\\[.5ex]
   & \prl      &  \acc a {x} . P 			&& (\text{replicated accept}) 		\\[.5ex]
   & \prl 	  &  \nil 										&& (\text{nil})				\\[.5ex] 
   & \prl 	  &  P \prl Q  								&& (\text{parallel})		\\[.5ex]
   & \prl 	  &  \newc {ab} P 					&& (\text{restriction}) 	\\[.5ex]
    & \prl      &  \kils a 									&& (\text{cancel})			\\[.5ex] 
   & \prl      &  \trycatch \spref P							&& (\text{catch})			\\
\end{array}
\end{gather*}
\end{center}
\caption{Syntax}\label{fig:syntax}
\end{figure}

\subsection*{Structural Congruence} 
With $\equiv$ we denote the least congruence on processes that is an equivalence relation, 
equates processes up to \aconv, 
satisfies the abelian monoid laws for parallel composition (with unit $\nil$), 
the usual laws for scope extrusion, and satisfies the axioms
below. (For the complete set of axioms with double binders, see~\cite{Vasconcelos2012}).
\begin{equation*}
\newc {ab} P \equiv \newc {ba} P 
\qquad 
\kils a \prl \kils a \equiv \kils a 
\qquad
\newc {ab} (\kils a \prl \kils b) \equiv \nil 
\qquad
\newc {ab} \kils a \equiv \nil 
\end{equation*}
The first axiom is needed for reduction;  
the second is needed for soundness; the remaining two are not strictly necessary 
but they allow to throw away garbage processes, specifically 
sessions that are fully cancelled.  

From now on, in all contexts (notably reduction, typing, proofs) we
shall consider processes up to structural equivalence; this is especially
useful in typing. 
Note that $\acc a x . P \not\equiv P \prl \acc a x . P$, \ie, we did not add 
the axiom for replication found in many presentations of $\pi$-calculus. 
We made this choice because adding the axiom would put to question the 
decidability of $\equiv$~\cite{milner:functions-as-processes}, 
and consequently of typing. 

\subsection*{Reduction}  

Do-catch contexts allow for possible exception handling.
\begin{equation*}
  H \grmeq [\,] \:\prl\: \trycatch {[\,]} P 
\end{equation*}
Notation $H[P]$ denotes the process obtained by filling the hole $[\,]$
in context $H$ with process $P$, as usual.

Reduction is defined in two parts: the standard rules
(Figure~\ref{fig:Rstd}), and the cancellation rules
(Figure~\ref{fig:Rcancel}).  First, recall that we work up to
structural equivalence, which means we do not explicitly state that
$P \equiv P' \R Q' \equiv Q \Rightarrow P \R Q$, but of course it
holds.
\begin{figure}
\begin{center}
{\small\begin{gather*}
\begin{array}{rcl}
\newc{ab}\blk{ \gctxi 1 {\sout {a} c . P}   \prl   \gctxi 2 {\sinp {b} {x} . Q} } 
   & \!\!\R\!\! & 
   \newc{ab} \blk{  P \prl Q\subs c x }  \hfill    \txsf{(R-Com)}
\\[.5ex] 
\newc{ab}\blk{ \gctxi 1 {\choiceI {a} {k} . P}   \prl    \gctxi 2 {\branchI {b} Q i I} }  
   & \!\!\R\!\! &   
    \newc{ab}\blk{  P \prl Q_k }  \hfill (k \in I)    \quad \txsf{(R-Bra)}
\\[.5ex]
\newc{ab}\blk{ \gctx { \req {a} {c} . P }     \prl  \acc {b} {x} . Q \prl R }  
   & \!\!\R\!\! &   
     \newc{ab}\blk{  P \prl Q\subs c x  \prl \acc {b} {x} . Q \prl R }  \quad \txsf{(R-Ses)}
\\[.5ex]
P \R Q 
  & \!\!\Rightarrow\!\! & 
     P \prl R \R Q \prl R   \hfill \txsf{(R-Par)}
\\[.5ex]
P \R Q 
  & \!\!\Rightarrow\!\! & 
     \newc {ab} P \R \newc {ab} Q   \hfill \txsf{(R-Res)}
\end{array}
\end{gather*}}%
\end{center}
\caption{Standard Reductions}\label{fig:Rstd}
\end{figure}
In standard reductions, the only notable point is that we discard any
do-catch handlers, since there is no cancellation, which explains why
the $H$-contexts disappear.  For example,
$\newc{ab}\blk{ \trycatch {\sout a c . P} { Q } \prl \sinp b {x} . R }
\R \newc{ab} \blk{ P \prl R\subs c x }$.
The type system ensures that it is sound to discard $Q$, since it
implements the same sessions as $P$ as well as the session on $c$.  On
the other hand, a cancellation activates a handler, which may provide
some default values to a session, completing it or eventually
re-throwing a cancellation.  For example, notice how $c$ appears in
the handler when $a$ is cancelled in
$\newc{ab} \blk{ (\trycatch {\sout a c . P} {\sout c 5.\kils c ) }
  \prl \kils b } \R \newc{ab} \blk{ \sout c 5.\kils c \prl \kils b }
\equiv \sout c 5.\kils c$.
    
%
\begin{figure}
\begin{center}
{\small \begin{gather*}
\begin{array}{rclr}
\newc{ab} \blk{ \acc a x . P \prl \kils b \prl R } 
& \R & \newc{ab} \blk{ \acc a x . P \prl R } & \txsf{(C-Acc)}
\\[.5ex]
\newc{ab} \blk{ \req {a} c . P   \prl  \kils {b} \prl R }   
     & \R &   
        \newc{ab} \blk{ P  \prl \kils {b}  \prl \kils c \prl  R }    & \txsf{(C-Req)}
\\[.5ex]
\newc{ab} \blk{ \sout {a} c . P   \prl  \kils {b} }   
     & \R &   
         \newc{ab} \blk{ P  \prl \kils {b}  \prl \kils c }    & \txsf{(C-Out)}
\\[.5ex]
\newc{ab} \blk{ \sinp {a} {x} . P \prl  \kils {b} }   
     & \R &    
        \newc{ab} \newc{xy} \blk{ P  \prl \kils {b} \prl  \kils {y} }   & \txsf{(C-Inp)}
\\[.5ex]
\newc{ab} \blk{ \choiceI {a} {k} . P   \prl \kils {b} }  
   & \R & 
      \newc{ab} \blk{ P \prl \kils {b} }    & \txsf{(C-Sel)}
\\[.5ex]
\newc{ab} \blk{ \branchI {a} P i I \prl \kils {b} }  
      & \R & 
          \newc{ab} \blk{ 
          P_k \prl \kils {b} } \hfill  \tsf{max}(I) = k 
              & \txsf{(C-Bra)}
\\[.5ex]
\newc{ab} \blk{ 
\trycatch \spref P \prl \kils {b} \prl R }   
     & \R &  
        \newc{ab} \blk{ P \prl \kils {b} \prl R } \quad \hfill \tsf{subject}(\spref) = a  
              & \txsf{(C-Cat)}
\end{array}
\end{gather*}}%
\end{center}
\caption{Cancellation Reductions}\label{fig:Rcancel}
\end{figure}

Our cancellation reductions are inspired by cut-elimination for
weakening in Proof Nets for Affine Logic
(see~\cite{Asperti:2002:ILAL}).  Specifically, $\kils a$ behaves like
a weakening (proof net) connective which consumes progressively
everything it interacts with (in logic this happens with cut).  For
example, using \tsf{(C{-}Inp)} we can perform
$\newc{ab}\blk{ \sinp a x . \sout x c \prl \kils b } \R
\newc{ab}\newc{xy} \blk{ \sout x c \prl \kils y \prl \kils b }$ and
then by \tsf{(C{-}Out)} we obtain
$\newc{ab}\newc{xy} \blk{ \sout x c \prl \kils y \prl \kils b } \R
\newc{ab}\newc{xy} \blk{ \kils b \prl \kils y \prl \kils c } \equiv
\kils c$.

In the cancellation of branching, \tsf{(C{-}Bra)}, we choose the
maximum index $k$ which exists given our assumption that index sets
are non-empty and totally ordered.  This is a simple way to avoid
non-determinism via cancellation, \ie, to ensure that cancellation
does not break confluence. 

In the rule \tsf{(C{-}Cat)}, we use a function $\tsf{subject}(\spref)$
which returns the subject in the prefix of $\spref$. This is defined
in the obvious way, \eg,
$\tsf{subject}(\sout a b . P) = \tsf{subject}(\req a b . P) = a$, and
similarly for the other prefixes $\rho$.  If $\spref$ happens to be a
request $\req a c . Q$, then $\kils b$ plays the role of an accept.
This explains why $\kils b$ remains in the result: like an accept, it
must be replicated to deal with possibly multiple requests in $P$ and
$R$.

The rule \tsf{(C{-}Acc)} is not strictly necessary for computation.
It simply reinforces the fact that a request does not cancel an
accept, a fact that may not be as obvious if we simply do not have a
reduction for this case. Moreover, it is important to define how
cancellation interacts with all constructors.

In the cancellation reductions \tsf{(C{-}Acc/Req/Cat)}, $R$ represents
the remaining scope of $a$ and $b$, so in the general case we should
have it also in \tsf{(C{-}Out/Inp/Sel/Bra)}.  However, the typing
system guarantees that in these cases both $a$ and $b$ are linear, and
therefore cannot appear elsewhere, so we preferred to keep the rules
simpler.  The same reasoning applies to the $R$ in the standard
reductions; only \tsf{(R{-}Ses)} needs it.

In rule  \tsf{(C{-}Inp)}, variable $y$ is not free in $P$, a fact that
results from the variable convention.


\section{Typing affine sessions}
\label{sec:typing}

This section introduces our notion of types and the typing system. It
motivates our choices and discusses the typing of the running example.

\subsection*{Types} 

The session types we use, shown in Figure~\ref{fig:types}, are based
on the constructs of Honda
\etal~\cite{honda.vasconcelos.kubo:language-primitives} with two
exceptions.  First, following Vasconcelos~\cite{Vasconcelos2012} we
allow a linear type to evolve into a shared type.  Second, following
the concept of Caires and
Pfenning~\cite{caires.pfenning:session-types-intuitionistic} we
decompose shared types into accept types $\aT T$ and request types
$\rT T$.
Technically, $\aT T$ corresponds to ${\boldsymbol \oc} T$ (\qt{of
  course $T\,$}) and $\rT T$ to ${\boldsymbol \wn} \dual T $ (\qt{why
  not $\dual T\,$}) from Linear Logic~\cite{Girard87}.  

\begin{figure}
\begin{center}
\begin{gather*}
\begin{array}{rclcr}
T        & \grmeq & \sessend 	    			&\qquad& (\text{nothing})			\\[4pt]
          & \prl      &  \send T . T         				    && (\text{output}) 		\\[4pt]
  		 & \prl      &  \recv T . T  			          && (\text{input}) 	\\[4pt]
   		& \prl      &  \sel \setT {T_i} {i} {I}   	&& (\text{selection})		\\[4pt]
   & \prl      &  \bra \setT {T_i} {i} {I}			&& (\text{branching}) 		\\[4pt]
   & \prl      &  \rT T     								&& (\text{request}) 		\\[4pt]
   & \prl 	  &  \aT T 										&& (\text{accept})				\\
\end{array}
\end{gather*}
\end{center}
\caption{Session Types}\label{fig:types}
\end{figure}

\subsection*{Duality} 

The two ends of a session can be composed when their types are {\em
  dual}, which is defined as an involution over the type constructors,
similarly to Linear Logic's negation except that $\sessend$ is
self-dual.%
\footnote{The expert might notice that logical negation suggests a
  dualisation of all components, \eg,
  $\dual {\send T . T'} \defeq \recv {\dual T} . \dual T' $ In fact
  the output type $\send T . T'$ and the request $\rT T$ hide a
  duality on $T$, effected by the type system, so everything is
  compatible.}
\begin{gather*}
\dual {\send {T_1} . T_2} \defeq \recv T_1 . \dual {T_2} 
\qquad
\dual {\recv T_1 . T_2} \defeq \send T_1 . \dual {T_2}  
\\[4pt] 
\dual {\sel \setT {T_i} {i} {I}} \defeq \bra \setT {\dual{T_i}} {i} {I} 
\qquad
\dual {\bra \setT {T_i} {i} {I}} \defeq \sel \setT {\dual{T_i}} {i} {I} 
\\[4pt] 
\dual{ \rT T } \defeq \aT T
\qquad
\dual{ \aT T } \defeq \rT T 
\qquad
\dual \sessend \defeq \sessend 
\end{gather*}

\subsection*{Interfaces (or typing contexts)}

We use $\Gamma,\Delta,$ and $\Theta$ to range over interfaces,
unordered lists of entries of the form $a\colon T$. 
We note that processes can have multiple
uses of $a \tp \rT T$, which corresponds to the logical principle of
contraction; this is the only kind of entry that can appear multiple
times in a context.
In this way, formation of contexts requires that $T=U=\rT V$ for
contexts $\Gamma,a\colon T$ when $a\colon U\in\Gamma$.
Henceforth, we assume all contexts are of this form, so that, whenever
we write $\Gamma,a\colon T$, then it must be the case that if~$a$ is
in~$\Gamma$, then its type is a request type equal to $T$.

To simplify the presentation, we identify interfaces up to permutations, 
so we do not need to define a type rule for the exchange of entries. 
We also use a pair of abbreviations: $\rT \! \Gamma$ stands for an
interface of the shape $a_1 \tp \rT {T_1}, \ldots,$
$a_n \tp \rT {T_n}$, and similarly, $\sessend\, \Gamma$ stands for an
interface $a_1 \tp \sessend, \ldots, a_n \tp \sessend$.

\subsection*{Typing rules}

Typing judgements take the form:
\begin{equation*}
\seq P \Gamma
\end{equation*}
meaning that process $P$ has interface $\Gamma$.

\begin{figure}[t]
\begin{gather*}
\trule{Out}
         { \seq{ P }{ \Gamma, a \tp T_2 } }
         { \seq{ \sout a b . P }{ \Gamma, a \tp \send T_1 . T_2, b \tp T_1 } }
\qquad
\trule{In}
         { \seq{ P }{ \Gamma, x \tp T_1, a \tp T_2 } }
         { \seq{ \sinp a {x} . P }{ \Gamma, a \tp \recv {T_1} . T_2 } }
\\
\trule{Sel}
         { \seq{ P }{ \Gamma, a \tp T_k } \\ k \in I }
         { \seq{ \choiceI a {k} . P }{ \Gamma, a \tp \sel \setT {T_i} {i} {I}  } }
\qquad
\trule{Bra}
         { \forall  i \in I \,.\, \seq{ P_i }{ \Gamma, a \tp T_i } \\ I \not = \emptyset }
         { \seq{\branchI a P i I  }{ \Gamma, a \tp \bra \setT {T_i} {i} {I}  } }
\\
\trule{Req}
         { \seq{ P }{ \Gamma } }
         { \seq{ \req a {b} . P }{ \Gamma, a \tp  \rT T, b \tp T } }
\qquad
\trule{Acc}
         { \seq{ P }{ \rT \Gamma, x \tp T } }
         { \seq{ \acc a {x} . P }{ \rT \Gamma, a \tp  \aT {T} } }
\\
\trule{Res}
         { \seq{ P }{ \Gamma_1, a \tp T } \\ \seq{ Q }{ \Gamma_2, b \tp \dual T } } 
         { \seq{ \newc{ab} ( P \prl Q ) }{ \Gamma_1, \Gamma_2 } }
\\
\trule{Contraction}
		 { \seq{P}{ \Gamma, a \tp \rT T, a \tp \rT T} }
		 { \seq{P}{ \Gamma, a \tp \rT T} }
\qquad
\trule{Nil}
         { \phantom{T} }
         { \seq \nil \emptyset }        
\qquad
\trule{Weak}
         { \seq P \Gamma }
         { \seq{ P }{ \Gamma, \rT \Delta, \sessend \,\Theta } }
\\
\trule{Catch}
         { \seq \spref {\Gamma, a \tp T} \\ \seq P {\Gamma} \\ \tsf{subject}(\spref) = a  }
         { \seq{ \trycatch \spref P }{ \Gamma, a \tp T } }
\qquad 
\trule{Cancel}
         { \phantom{\norequests(T)} }
         { \seq{ \kils a }{ a \tp T } }
\end{gather*}
\caption{Affine Session Typing}\label{fig:typing}
\end{figure}

The typing rules are presented in Figure~\ref{fig:typing}.  We focus
on some key points, noting that a rule can only be applied if the interface of the conclusion 
is well-formed.  

In \tsf{(Out)}, 
an output $\sout a b . P$
records a conclusion $b \tp T_1$, so in fact it composes against
the dual $b \tp \dual {T_1}$. Therefore $\send T_1 . T_2$ really means to send
a name of type  $\dual {T_1}$, which matches with the dual input. The same
reasoning applies to requests; see \tsf{(Req)}. 
In \tsf{(Res)} we split the process so that each part implements one
of the ends of the session. This is inspired by Caires and
Pfenning~\cite{caires.pfenning:session-types-intuitionistic} which
interprets sessions as propositions in a form of Intuitionistic Linear
Logic; the notion of \qt{cut as composition under name restriction}
comes from Abramsky~\cite{Abramsky93CILL}.

In \tsf{(Cancel)}, $\kils a$ can be given any type.   
A \docatchk{} process is typed using rule \tsf{(Catch)}, as follows: if
$\spref$ is an action on $a$ and has an interface $(\Gamma, a \tp T)$,
then the handler $P$ will implement $\Gamma$, \ie, all sessions of
$\spref$ \emph{except} for $a \tp T$ which has been cancelled. 
The rule is sound, since no session
is left unfinished, irrespectively of which process we execute, $\spref$ or $P$. 
Notice that, if $T = \rT U$, we may have more occurrences of
$a \tp T$ in $\Gamma$, because of contraction.  
We made the choice to
allow this, since it does not affect any property.%
\footnote{On the other hand, adding a premise
  $a \not\in\tsf{dom}(\Gamma)$ in \tsf{(Catch)} would cause problems
  with subject reduction.  For example
  $\newc{cd}\blk{ \sout ca \prl \sinp d x . \trycatch {\req a y }{
      \req x z }}$ would be typable, but it reduces to
  $\trycatch {\req a y }{ \req a z }$ which is not typable.}

\subsection*{Typing the book purchase example}

It is easy to verify that the examples in Section~\ref{sec:example}
are well-typed.  For the \tsf{Buyer} we obtain the following
sequent 
\begin{equation*}
  \seq{
    \tsf{Buyer}
  }{
    \textit{ccard} \tp \tsf{string}, \textit{seller}_1 \tp \rT {T_1} 
  }
\end{equation*}
where
$T_1$ abbreviates type $\recv {\tsf{string}} . \send {\tsf{double}} . 
         \bra \lbrace 
                \tsf{buy} \tp \recv {\tsf{string}} . T_2,
                \tsf{cancel} \tp \sessend 
                \rbrace$
and $T_2$ stands for $\sel \set{ \tsf{accepted} \tp \sessend, \tsf{rejected} \tp \sessend }$.
The behaviour of $b$ inside process \tsf{Buyer} is described by $\dual{T_1}$. 

For the \tsf{Bank} we obtain the following sequent
\begin{equation*}
  \seq{
    \tsf{Bank}
  }{
    \textit{bank}_2 \tp \aT {T_3}
  }
\end{equation*} 
with $T_3 =  \recv {\tsf{double}} . 
             \wn ( \recv {\tsf{string}} . T_2 ). 
                         \send { T_2 } . 
                         T_2$.

Finally, for the \tsf{Seller} we have
\begin{equation*}
  \seq{
    \tsf{Seller}
  }{
    \textit{bank}_1 \tp \rT {T_3}, \textit{seller}_2 \tp \aT {T_1} 
  }
\end{equation*}

Interestingly, no type structure is needed for the affine adaptations: 
cancellation is completely transparent.  The variation of \tsf{Seller}
with an added \docatchk{}:
\begin{equation*}
\trycatch {\req {\textit{bank}_1} {\,k'}} {\req {\textit{paymate}} {\,k'}}
\end{equation*}
will simply need $\textit{paymate} \tp \rT {T_3}$ in its interface,
\ie, with a type matching that of $\textit{bank}_1$, but the original
\tsf{Seller} can also be typed in the same way by using weakening to
add an extra assumption to the context.  Similarly, \tsf{BuyerMsg} has
the same interface as \tsf{Buyer}, except that it must include
$\textit{print} \tp \rT {\tsf{string}}$, and again the two processes
can be assigned the same interface by weakening, if needed.  Processes
\tsf{CheckPriceA}, \tsf{CheckPriceB}, and \tsf{BuyerCancel} can typed
under contexts containing $\textit{seller}_1 \tp \rT {T_1}$, exactly
as in process \tsf{Buyer}.
 
Finally, as we shall see next affinity does not destroy any of the
good properties we expect to obtain with session typing.

\subsection*{Typing modulo structural equivalence}

Since we consider processes up to $\equiv$, we have that $P \equiv P'$
and $\seq{P'}\Gamma$ implies $\seq P \Gamma$.  This possibility is
suggested by Milner~\cite{milner:functions-as-processes} and used by
Caires and
Pfenning~\cite{caires.pfenning:session-types-intuitionistic}.
It is necessary because associativity of \qt{$\,|\,$} does not
preserve typability, for example $\newc{ab}\blk{ P \prl (Q \prl R)}$
may be untypable in the form $\newc{ab}\blk{ (P \prl Q) \prl R}$.
This applies also to scope extrusion $\newc{ab}(P \prl Q)$ and
$\newc{ab}P \prl Q$.

\subsection*{Implicit and explicit weakening} 

The rules \tsf{(Weak)} and \tsf{(Cancel)} implement a form of
weakening.  In the case of \tsf{(Weak)}, this weakening is
\emph{implicit}, in the sense that we extend an interface without
adding any behaviour at the process level.  Specifically, to close a
session (introducing $a \tp \sessend$) or to record that a service is
invoked (with $a \tp \rT T$) weakening is standard.  On the other
hand, \tsf{(Cancel)} introduces an \emph{explicit} weakening, since we
maintain the logical concept of a device to perform it, and this is
why, as we shall see, sessions do not get stuck.

\subsection*{Deriving the \tsf{Mix} rule}
A common situation is when we want to compose independent processes, 
\ie, processes that do not communicate. 
This corresponds to the introduction of Girard's \qt{Mix} rule, 
which is derived in our system as follows: 
\begin{gather*}
         \trule{Mix}
         { \seq{ P }{ \Gamma_1 } \\ \seq{ Q }{ \Gamma_2 }  }
         { \seq{ P \prl Q }{ \Gamma_1, \Gamma_2 } }
         \;\;
         \defeq
         \;\;
         \dfrac{  \dfrac{ \seq P  {\Gamma_1} } 
                                { \seq P {\Gamma_1, a \tp \sessend} } \; \tsf{(Weak)} 
                      \qquad
                     \dfrac{ \seq Q {\Gamma_2} } 
                                { \seq Q {\Gamma_2, b \tp \sessend} } \; \tsf{(Weak)}
                  }
                  { \seq{P \prl Q}{\Gamma_1, \Gamma_2}} \;\tsf{(Res)} 
\end{gather*}
Recall that whenever $a,b$ are not free in $P$ and $Q$, we have
$P \prl Q \equiv \newc {ab} \blk{ P \prl Q }$.

\subsection*{The need for double binders} 

In our system, the terms
$\newc{ab}\blk{ \req a c . P \prl \kils b \prl R }$ and
$\newc{ab}\blk{ \req a c . P \prl \kils a \prl R }$ behave
differently.  In particular, there is a reduction of the first one,
using \tsf{(C{-}Req)}, but not of the second one where both the
request and the cancellation are on the same endpoint, namely $a$.
Both terms are typable, but with different derivations: for the first
we need to use \tsf{(Res)} on $a \tp \rT T$ and $b \tp \aT T$, while
on the second we apply \tsf{(Mix)} and possibly \tsf{(Contraction)}
for $a \tp \rT T$.  However, in a system with the standard scope
restriction from $\pi$-calculus there are no separate endpoints and
therefore the two terms above are identified as
$\newc{a}\blk{\req {a} {c} . P \prl \kils {a} \prl R\subs a b}$.

Can we apply a cancellation reduction? The two terms above have
different dynamics, so there is no solution once we identify them.
From the perspective of the typing system, we would not know any more
if the type of $\kils a$ in the last term is $a \tp \rT T$ or
$a \tp \aT T$, since there can be different type derivations for each
possibility.  In the extended abstract of this
work~\cite{Mostrous2014} we did not employ double-binders, and as a
result we had to deal with this ambiguity. 
Our solution was to impose a rather severe condition on typing,
roughly that $\seq{\kils a}{a \tp T}$ only if $T$ does not contain a
subexpression $\rT T'$. In this way we excluded one of the two
possible cases.  The introduction of double binders resolves such
ambiguities without any restriction on the shape of types.


\section{Properties}
\label{sec:properties}

This section discusses the tree main properties of our system:
soundness, confluence, and progress.

\begin{thm}[Subject Reduction]
  \label{thm:sr}
  If $\seq P \Gamma$ and $P \R Q$ then $\seq Q \Gamma$.
\end{thm}
\begin{proof}
See Appendix~\ref{annex:thm:sr}.
\end{proof}

\begin{thm}[Diamond property]
  \label{thm:diamond}
  If $\seq P \Gamma$ and $Q_1 \longleftarrow P \R Q_2$ then either
  $Q_1 \equiv Q_2$ or $Q_1 \longrightarrow R \longleftarrow Q_2$.
\end{thm}
\begin{proof}
  The result is easy to establish, since the only critical pairs would
  arise from multiple requests to the same replication or to the same
  cancellation.  However, even in that case the theorem holds because:
  a) replications are immediately available and functional
  (\emph{uniform availability}); b) cancellations are persistent.

Possible critical pairs arise when two reductions overlap (they both use a common sub-process).
Let us consider a process  
$\newc{ab}\blk{ H_1[ \req a {c_1} . P_1 ] \prl H_2[ \req a {c_2} . P_2 ] \prl \acc b {x} . P_3 }$. 
There are two critical pairs, depending on which request reacts first, resulting in either 
the process $\newc{ab} \blk{ P_1 \prl  P_3\subs {c_1}{x} \prl H_2[ \req a {c_2} . P_2 ] \prl \acc b {x} . P_3 }$ 
or alternatively to the process  
$\newc{ab} \blk{ P_2 \prl  P_3\subs {c_2}{x} \prl H_1[ \req a {c_1} . P_1 ] \prl \acc b {x} . P_3 }$. 
But both reduce in one step to 
$\newc{ab} \blk{ P_1 \prl P_3\subs {c_1}{x} \prl P_2 \prl P_3\subs {c_2}{x} \prl \acc b {x} . P_3 }$ . 

Another possibility is
$\newc{ab} \blk{ H_1[ \req a {c_1} . P_1 ] \prl H_2[ \req a {c_2} . P_2 ] \prl \kils b }$.  Again, the
critical pair is trivial.  

In both cases above, the same reasoning applies if there are more than two possible reductions.  
\end{proof}

The above strong confluence property 
indicates that our sessions are completely deterministic, 
even considering the possible orderings of requests. 

\subsection*{Progress}

Our contribution to the theory of session types is
\textit{well-behaved affinity}, in the sense that we can guarantee
that any session that ends prematurely will not affect the quality of
a program.  Indeed, if we simply allowed unrestricted weakening, for
example by a type rule $\seq \nil \Gamma$ as done
in~\cite{giunti2013algorithmic}, but without any cancellation
apparatus at the language level, it would be easy to type a process
such as $\newc{xz} \blk{ \newc{ab} \sout a x . P \prl \sinp z y . Q }$ 
and clearly not
only $a$ (and everything in $P$) but also $z$ (and everything in $Q$)
would be stuck for ever.  In this section we prove that this never
happens to a well-typed process.

The next notion identifies the set of characteristic processes,
conducting a session $T$ on a channel $a$.

\begin{defi}[Atomic Action]
The atomic actions are defined by the grammar below.
\begin{equation*}
 \aact \grmeq 
               \gctx \spref 
  \:\prl\:  \acc a {x} . P
  \:\prl\:  \kils a 	
\end{equation*}
\end{defi}

\begin{defi}[Active and Inactive Processes]
  \label{def:inact}
  We say that a process $P$ is \emph{inactive}, written $\inact P$, if
  it belongs to the set
  \begin{equation*}
    \{ P \prl \seq{P}{\Gamma} \text{ and } P \equiv \newc{\vec {cc'},
      a_1a_1',\cdots , a_n a_n'}(\acc {a_1} {x_1}.P_1 \mid\dots\mid \acc {a_n} {x_n}.P_n\}
  \end{equation*}
  Otherwise $P$ is called \emph{active}, written $\acti P$.
\end{defi}

The reason for distinguishing inactive processes is that they are
typable (and naturally emerge) because an accept can be used zero or
more times.  For example the process
$\newc{ab}\blk{ \req a {c}. \nil \prl \req a {r}. \nil \prl \acc b {x}
  . \nil }$ is active and can be typed with interface
$(c \tp \sessend, r \tp \sessend)$, but after two steps it reduces to
the inactive process $\newc{ab} \acc b {x} . \nil$.  To type the last
process, we need to transform it to
$\newc{ab} \blk{ \acc b {x} . \nil \prl \nil }$ and apply
\tsf{(Nil,Weak)} so as to obtain $\seq{\nil}{a \tp \rT \sessend}$,
then \tsf{(Res)} to obtain the result.

\begin{defi}[Observation Predicate]
  \label{def:obs}
  We write \barb{a}{P} whenever
  $P \equiv \newc{\vec {bb'}}\blk{ \aact \prl P' }$ with
  $\tsf{subject}(\aact) = a$ and $a \not\in\vec b, \vec {b'}$.  We
  write \wbarb{a}{P} whenever $P \R^* P'$ and \barb{a}{P'}, for some
  $a$.
\end{defi}

Below we write $P \notR$ if there does not exist $Q$ such that $P \R Q$. 
\begin{lem}
  \label{lem:prog}
  If $\seq P \Gamma$ and $P \notR$ and $\acti P$, then
  $\Gamma = \Gamma', a \tp T$ and $\barb a P$.
\end{lem}
\begin{proof}
By rule induction on the first premise.
\end{proof} 

The assumption $P \notR$ is necessary, since there are active
processes that reduce to inactive processes, such as
$\newc{cc'}\blk{\sout c a \prl \sinp {c'} x}$ or the example under
Definition~\ref{def:inact}.
Our formulation of the Lemma~\ref{lem:prog} is close in principle to
the approach of Dezani-Ciancaglini \etal~\cite{DLY07} and is similar
to Caires and Pfenning's Lemma~4.3
~\cite{caires.pfenning.toninho:mscs}, except that we use the shape of
the process (having a top-level action on $a$) instead of the ability
of the process to perform a labelled transition $P \stackrel{a}{\R}$.

The next notion identifies the set of characteristic processes of the
type of any given session on $a$, following the terminology and main
idea of Coppo \etal~\cite{CDYP2015}.

\begin{defi}[Characteristic Process]\label{def:characteristic-proc}
  The set of processes \cata{a \tp T} is defined, for each channel end
  $a$, inductively on the type $T$. We take $b, b'$ to be different from 
  $a$.
\begin{align*}
\cata{a \tp \sessend} & \defeq  \set{ \nil } 
\\
\cata{ a \tp \send T_1 . T_2 } & \defeq 
     \set{ \newc {bb'} \blk{ \sout a {b'} . P \prl Q } \setsep P \in \cata{a \tp T_2}, Q \in \cata{b \tp \dual{T_1}}  }
\\
\cata{ a \tp \recv T_1 . T_2 } & \defeq 
     \set{ \sinp a x .\blk{  P \prl Q } \setsep P \in \cata{a \tp T_2}, Q \in \cata{x \tp T_1}  }
\\
\cata{ a \tp \sel \setT {T_i} {i} {I} } & \defeq 
     \cup_{i \in I} \set{ \choiceI a {i} . P \setsep P \in \cata{a \tp T_i} } \quad (I \not = \emptyset)
\\
\cata{ a \tp \bra \setT {T_i} {i} {I} } & \defeq
     \set{ \branchI a P i I \setsep P \in \cata{a \tp T_i} }
\\
\cata{ a \tp \rT T_1 } & \defeq
     \set{ \newc {bb'} \blk{ \req a {b'}  \prl P } \setsep P  \in \cata{b \tp \dual{T_1}}  }
\\
\cata{ a \tp \aT T_1 } & \defeq
     \set{ \acc a {x} . P  \setsep P  \in \cata{x \tp T_1}  }
\end{align*}
\end{defi}

The clauses \cata{ a \tp \send T_1 . T_2 } and \cata{ a \tp \recv T_1
  . T_2 } identify sets of characteristic processes which are
\emph{non-interleaving}.  This means that they do not mix different
sessions within the same prefix sequence.  Indeed, the reason is that
$P$ must continue its evolution on $a$ independently from $Q$ which
instead works on $b$ or $x$.  Moreover, all characteristic processes
are \emph{complete}, meaning that they provide communication actions
for the whole session as defined by the type.  For example, any
process in
$\cata{ a \tp \send {\sessend} . \recv {\sessend} . \sessend }$ will
perform an output on $a$ and then an input also on $a$. On the other
hand, there are processes such as $\sout a b . \sout c a$ which can be
assigned
$\Gamma, a \tp \send {\sessend} . \recv {\sessend} . \sessend$ but do
not complete the session on $a$.
\begin{prop}\leavevmode
  \label{prop:char-proc}
  \begin{enumerate}[(a),ref={\alph*}]
  \item For all $a$ and $T$, \cata{a \tp T} is
    non-empty; \label{prop:nonempty-char-proc}
  \item For all $P \in \cata{a \tp T}$ we have $\seq{P}{a \tp
      T}$; \label{prop:char-proc-interface}
  \item For all $P \in \cata{a \tp T}$, either $T = \sessend$ or
    \barb{a}{P}.\label{prop:char-proc-observe}
  \end{enumerate}
\end{prop}
\begin{proof}
  (a) Immediate from the definition.
  (b) Easy induction on $T$.  
  (c) In all cases except for $T = \sessend$, the characteristic processes $P \in \cata{a \tp T}$ are of the form $\alpha$
  or $\newc {bb'} (\aact \prl R)$ with $\tsf{subject}(\aact)=a$, and therefore we have \barb{a}{P}.
\end{proof}

We can now claim a standard progress result, which guarantees that
active processes can always perform some interaction.

\begin{cor}[Progress]
  \label{cor:prog}
  If $\seq P \Gamma$ and $P \notR$, then either $\inact P$ or there
  exists $Q$, $a$, $a'$, $\Delta$, $\Theta$ with $\seq Q \Delta$ and
  $Q \notR$, such that $\seq {\newc{a a'} (P \prl Q)}{\Theta}$ and
  $\newc{a a'} (P \prl Q) \R$.
\end{cor}
\begin{proof}
  We focus on the interesting case, which is when the process $P$ is
  active.  From the assumptions and from Lemma~\ref{lem:prog} we know
  that $\seq{P}{\Gamma, a \tp T}$ and $\barb a P$; notice that
  $T \not = \sessend$, since this contradicts $\barb a P$.  But then
  we can compose with a characteristic process
  $Q \in \cata{ a' \tp \dual T }$ obtaining
  $\seq{\newc{a a'}\blk{ P \prl Q }}{\Gamma}$, and clearly
  $\newc{a a'}\blk{ P \prl Q } \R$ because both $P$ and $Q$ (see
  Proposition~\ref{prop:char-proc}(\ref{prop:char-proc-observe})) are
  ready to react on $a$ and $a'$, respectively, having actions of dual type.
\end{proof}
 
At first sight, Corollary~\ref{cor:prog} may seem to allow some
deadlocks.  For example, it might seem that we could have a deadlocked
sub-process $R$ in parallel to a request $\reqi a 1 b . P'$, then we
could always apply a parallel composition with a forwarder
$\acc {a_i'} x . \reqi a {i+1} x$ (as the $Q$ in
Corollary~\ref{cor:prog}) and there would always be a reduction step.
However, the sub-process $R$ would need to be typed as part
of the larger derivation, but then we arrive at a contradiction: 
by Corollary~\ref{cor:prog}, $R$ must
be able to perform an action given a suitable context, which
contradicts the assumption that it is deadlocked. 
To prove this formally, it suffices to show that applications of 
\tsf{(Res)} do not inhibit any action except the restricted one, 
which is easy to establish. 

\subsection*{Discussion}

Finally, it can be shown that typed processes are strongly
normalising, which is not so surprising since we followed closely the
logical principles of Affine Logic.  This can be shown by a small
adaptation of the standard method~\cite{Girard87}, first by giving an
interpretation of types based on biorthogonals, then by strengthening
the induction hypothesis using a notion of reducibility (a contextual
test for normalisation), and finally by making use of
Theorem~\ref{thm:diamond} to obtain strong normalisation from weak
normalisation.

Note that Progress (Corollary~\ref{cor:prog}) is in a sense more
important, for two reasons: first, a system without progress can still
be strongly normalising, since blocked processes are by definition
irreducible; second, practical systems typically allow recursion, and
in that case the progress property (which we believe can be
transferred without surprises to this setting) becomes much more
relevant.


\section{Related work and future plans}
\label{sec:related-work}

We divide our discussion on the related work in three parts: relaxing
linearity in session types, dealing with exceptional behaviour, 
and logical foundations. 

The study of language constructs for exceptional behavior (including
exceptional handling and compensation handling) has received
significant attention; we refer the reader to a recent
overview~\cite{ferreira.etal:advanced-mechanisms-service-combination-transactions},
while concentrating on those works more closely related to ours.
Carbone \etal are probably the first to introduce exceptional
behaviour in session
types~\cite{carbone.etal:structured-interactional-exceptions}. They do
so by extending the programming language (the $\pi$-calculus) to
include a \tsf{throw} primitive and a try-catch process. The language
of types is also extended with an abstraction for a try-catch block:
essentially a pair of types describing the normal and the exceptional
behaviour. The extensions allow communication peers to escape, in a
\emph{coordinated} manner, from a dialogue and reach another point
from where the protocol may progress.
Carbone~\cite{carbone:session-choreography-exceptions} and Capecchi
\etal~\cite{capecchi.etal:global-escape-in-multiparty-sessions} port
these ideas to the multi-party setting.
Hu \etal present an extension of multi-party session types that
allow to specify conversations that may be
interrupted~\cite{hu.etal:practical-interruptible-conversations}. Towards
this end, an \tsf{interruptible} type constructor is added to the
type language, requiring types that govern conversations to be
designed with the possible interrupt points in mind.
In contrast, we propose a model where programs with and without
exceptional behaviour are governed by the same (conventional) types,
as it is the norm in functional and object-oriented programming
languages.

Caires \etal proposed the conversation
calculus~\cite{vieira.etal:the-conversation-calculus}. The model
introduces the idea of conversation context, providing for a simple
mechanism to locally handle exceptional conditions. The language
supports error recovery via \tsf{throw} and \tsf{try}{-}\tsf{catch}
primitives. In comparison to our work, no type abstraction is proposed
since their language is untyped, moreover errors are caught in a
context and do not follow and destroy sessions as in our work, and
finally their exception mechanism does not guarantee progress.

Contracts take a different approach by using process-algebra
languages~\cite{castagna.etal:theory-contracts-web-services} or
labeled transition
systems~\cite{bravetti.zavattaro:contract-based-web-services} for
describing the communication behaviour of processes. In contrast to
session types, where client-service compliance is given by a symmetric
duality relation, contracts come equipped with an asymmetric notion of
compliance usually requiring that a client and a service reach a
successful state. In these works it is possible to end a session
(usually on the client side only) prematurely, but there is no
mechanism equivalent to our cancellation, no relationship with
exception handling, and no clear logical foundations.

Caires and Pfenning gave a Curry-Howard correspondence relating
Intuitionistic Linear Logic and session types in a synchronous
\pic~\cite{caires.pfenning:session-types-intuitionistic}.  To avoid
deadlocks, we follow the same approach and impose that any two
processes can communicate over at most a single session; in our case
this in ensured in \tsf{(Res)}.  Therefore, both systems
reject processes of the shape
$\newc{a a' , b b'}\blk{ \sinp a x . \sout b c \prl \sinp {b'} y . \sout {a'} r }$ which are stuck.
This restriction is founded on logical cut and was first introduced
into a process algebra by Abramsky~\cite{Abramsky93CILL}.  
Caires and Pfenning~\cite{caires.pfenning:session-types-intuitionistic} 
achieve a progress property similar to
Corollary~\ref{cor:prog} in our work, but only in the more rigid setting of linearity. 
Adding affinity to their system, without a mechanism similar to our cancellation, 
would allow sessions to get stuck. 
Our formulation allows to type more processes than Linear Logic
interpretations, such as \tsf{BuyerCancel} from
Section~\ref{sec:intro} and the alternative form of choice shown in
Section~\ref{sec:typing}.  Moreover, to our knowledge our work is the
first logical account of exceptions in sessions, based on an original
interpretation of weakening.  Finally, Propositional Affine Logic is
decidable, a result by Kopylov~\cite{kopylov.decidability-lal}, so we
face better prospects for type inference.

As part of future work, we would like to develop an algorithmic typing
system, along the lines of~\cite{Vasconcelos2012}.  We also believe it
would be interesting to apply our technique to multiparty session
types~\cite{HondaK:mulast} based on Proof
Nets~\cite{multipartynets14}. 
Finally, we plan to study the Curry-Howard correspondence with Affine Logic 
in depth, and examine more primitives that become possible with our mechanism.


\section*{Acknowledgement}

We would like to thank the anonymous reviewers of COORDINATION'14 and
of LMCS, and also Nobuko Yoshida, Hugo Torres Vieira, Francisco
Martins, and the members of the \textsc{Gloss} group at the University
of Lisbon, for their detailed and insightful comments.


\bibliography{sessions}
\bibliographystyle{plain}

\newpage
\appendix
\section{Subject Reduction}
\label{annex:thm:sr}

\subsection*{Notation for type derivations}

Let $\tyderT{}{\seq P \Gamma}$ stand for a typing derivation 
with conclusion $P$ and interface $\Gamma$. 
To ease the notation, sometimes we ignore the process, 
writing $\tyderT{}{\Gamma}$, 
or even the interface, writing simply $\tyder$. 

We let $\tyderHs{}{\Delta}$ stand for a derivation with a hole that
contributes an interface $\Delta$.
We make a simplification and require the hole to be unguarded in $\tyder$, 
which means that it does not appear inside a sub-derivation 
$\tyderT{1}{\aact}$ of $\tyder$. 
To put simply, we do not consider holes in sub-derivations guarded by a prefix. 

\begin{defi}[Derivation Composition]\label{def:dercomp}
We denote by $\tyderC{1}{\tyder_2}{\Gamma, \Theta}$ 
the derivation obtained from 
$\tyderH{1}{\Delta,\Lambda}{\Gamma,\Lambda}$ by the placement of 
$\tyderT{2}{\Delta, \Theta}$ in the hole, 
assuming the following conditions: 
\begin{enumerate}
\item no element of $\Theta$ is bound in 
$\tyder_1$ (but elements of $\Delta$ can be bound); 
\item $\Gamma, \Theta$ is well-formed (as is $\Delta, \Theta$, by assumption).
\end{enumerate}
\end{defi}
The point of these conditions is to guarantee that compositions are
sound, and specifically they constrain the type of derivation that can
be plugged-in a hole allowing us to know the final interface; this is
very useful in Subject Reduction.  
The part $\Lambda$ is what passes from the inner context to the conclusion, 
which means it is unused in the derivation. 
The idea is that we can substitute it with some $\Theta$ that also passes 
to the conclusion, 
under suitable assumptions, namely that no name in $\Theta$ is bound in 
the derivation context. 
This allows to place in the hole a derivation with different (typically larger) 
interface, which is what happens during our Subject Reduction proof.

After we insert a derivation in place of the hole, 
we thus obtain $\tyderC{1}{\tyder_2}{\Gamma, \Theta}$.

\subsection*{Remark}
If we wanted to provide a detailed definition of derivations with
holes---which we think is not necessary--- we would use an extended
version of the typing system with the extra axiom
$\overline{\seq \cdot \Delta}$ and would require that it appears at
most once in any derivation $\tyder$.  (And when it does appear we
obtain $\tyderHs{}{\Delta}$, otherwise we have a normal derivation.)
 
Regarding the restriction to unguarded holes in a derivation, we do
not need the more general form where the hole can be anywhere, since
parts under prefix do not reduce and we never need to manipulate them.
Also, our restriction simplifies the notion of placing a derivation
with different interface in a hole, which we need to use extensively.
Specifically, if the hole was allowed to appear inside a
sub-derivation with restrictions on the resulting interface --- such
as in the premises of \tsf{(Acc)} or \tsf{(Bra)} --- then an arbitrary
interface added to them ($\Theta$ above) would not guarantee that we
obtain a well-formed derivation, \ie, this plugging-in of a derivation
could be unsound.

\begin{prop}
  \label{prop:subder}
  If $\tyderH{1}{\Delta,\Lambda}{\Gamma,\Lambda}$ and the conditions
  of instantiation detailed previously (see Def.~\ref{def:dercomp})
  are respected for some $\tyderT{2}{\Delta, \Theta}$, 
then $\tyderC{1}{\tyder_2}{\Gamma, \Theta}$. 
\end{prop}
\begin{proof}
This is easy to establish by induction on the structure 
of $\tyder_1$.
\end{proof}

\begin{lem}[Atomic Action Sub-derivation]
  \label{lem:atomic-typa}
  If $\tyderT{ }{\seq P \Gamma}$ and
  $P \equiv \newc{\overrightarrow{ b b'}} ( \aact \prl Q )$ with
  $\tsf{subject}(\aact) = a$, then \tyder\ is of the shape
  $\tyder_1[ \tyderT{2}{\seq{\aact}{\Delta, a \tp T}} ]$.
\end{lem}

To put this into English, when a communication action $\aact$ appears 
at the top level of a well-typed process $P$, then any derivation $\tyder$ 
of $P$ 
will contain a sub-derivation with $\aact$ as its conclusion. 
Notice that we cannot guarantee that the last rule corresponds to 
the introduction of $\aact$.
We give a simple example:
\[
	P = \newc{c c'}(\,   \underbrace{\sout a b .\sout {c'} r}_\aact  \prl \sinp c x . Q   \,)
\]%
If the above is well-typed, the last rule is \tsf{(Res)} on endpoints $c$ and $c'$  
and cannot be \tsf{(Out)} on channel $a$. 
This situation is obvious in Affine and Linear Logic, 
and in particular it is clear that without commutative reductions  
the instance of \tsf{(Res)} --- which corresponds to a \qt{cut} 
in Affine Logic --- cannot go \qt{up} in the derivation. 
In particular we would need 
$P \equiv \sout a b .\newc{c c'}(\,  \sout {c'} r  \prl \sinp c x . Q   \,)$.\footnote{
This becomes unmanageable for more complex processes with multiple cuts, unless if type derivations 
are manipulated (see for example~\cite{Wadler12propsess}).}
This lemma is needed to identify and manipulate the derivations 
obtained during the proof of Subject Reduction. 
Specifically, when dealing with a redex inside a larger derivation, 
we can use this inversion-like 
lemma to obtain the sub-derivations for the two communications, 
compose them, use Lemma~\ref{lem:ssr} (Small Subject Reduction, defined shortly), 
and finally recompose the derivation for the contractum 
into the original derivation. 

\begin{prop}[Sub-interface Compositionality]
\label{prop:iorder}
If
$\tyderC{1i}{ \tyderT{2i}{\Delta_i, a_i \tp T_i} }{ \Gamma_i, a \tp T_i
}$
($i\in\set{1,2}$) and $\Gamma_1,\Gamma_2$ is well-formed, then
$\Delta_1,\Delta_2$ is well-formed.
\end{prop}
\begin{proof}
A simple explanation of this proposition is that, if two derivations can be composed 
using \tsf{(Res)}, as is the case when $T_1 = \dual {T_2}$, 
then their sub-derivations can also be composed 
without conflicts.  

In more detail, recall that well-formedness means that only elements of the shape 
$a \tp \rT T$ can appear multiple times in an interface. 
In general, the $\Gamma_i$ can have a larger interface than the $\Delta_i$ 
since other parts of the derivation can extend it, which is no problem. 
The only obstacle to the well-formedness of $\Delta_1,\Delta_2$  
would be the case where the $\Delta_i$ have a larger domain than the corresponding 
$\Gamma_i$, since in that case the well-formedness of $\Gamma_1,\Gamma_2$ does 
not imply that of $\Delta_1,\Delta_2$. 
Now, let us consider the possible differences between $\Delta_i$ (the sub-derivation interface) 
and $\Gamma_i$.  
Because of weakening, $\Gamma_i$ can have a larger domain 
(set of names). Because of contraction, $\Gamma_i$ can have lesser copies 
of some $a \tp \rT T$, but in this case the domain is the same. 
Session constructors change the types so that $\Delta_i$ can contain 
some $a \tp T$ which appears as $a \tp T'$ in $\Gamma_i$, 
so again the domain is the same. 
Finally, some part of $\Delta_i$ can be closed (by scope restriction) and in that case 
--- only in that case --- the domain of $\Delta_i$ can be larger than that of $\Gamma_i$. 
However, by the variable convention all bound names are distinct from each other and 
from all free names, and therefore there is no conflict when composing 
$\Delta_1, \Delta_2$. Therefore, the sub-derivations have interfaces that 
can be composed, i.e., the result is well-formed as required. 
\end{proof}

\begin{lem}[Free Variables]
  \label{lem:fv}
  If $a$ is free in $P$ and $\seq P \Gamma$ then $a\tp T \in \Gamma$.
\end{lem}

For the purpose of the Substitution Lemma, we denote by
$\vec x\colon T$ the \emph{non-empty} context
$x\colon T, \dots, x\colon T$. According to context formation, if $x$
occurs twice in the context, then it must be the case that~$T$ is a
request type.

\begin{lem}[Substitution]
  \label{lem:sub}
  If $\seq{P}{\Gamma,\vec x\colon T}$ and $x\notin\Gamma$ then
  $\seq{P\subs b x}{\Gamma,b\colon T}$.
\end{lem}
\begin{proof}
  The proof is by rule induction on the first premise. Most cases
  feature two or three subcases depending on~$x$ being equal or
  different to names $a$ and $b$ in the typing rules. Extra subcases
  may also be needed depending on~$T$ being a request type or not. We
  illustrate one representative case: when derivation ends with the
  \tsf{(Catch)} rule.

  \Case $a=x$ and $T$ is not a request type.
  We have $\seq{\trycatch\spref P}{\Gamma, x \tp T}$ and
  $x\notin\Gamma$.
  By rule inversion we obtain $\seq \spref {\Gamma,x\tp T}$ and
  $\seq P {\Gamma}$ and $\tsf{subject}(\spref) = x$.
  By induction, we get $\seq{\spref\subs bx}{\Gamma, b\tp T}$.
  We are now in a position to apply rule \tsf{(Catch)} to obtain
  $\seq{\trycatch {\spref\subs bx} P}{\Gamma, b\tp T}$.
  To complete the case we note that Lemma~\ref{lem:fv} gives us
  that~$b$ is not free in~$P$, hence
  $\trycatch {\spref\subs bx} P = (\trycatch \spref P)\subs bx$.

  \Case $a=x$ and $T$ is a request type.
  We have
  $\seq{\trycatch \spref P}{\Gamma, \vec x\tp T,x \tp T}$
  and $x\notin\Gamma$. By rule inversion we obtain
  $\seq \spref {\Gamma, \vec x\tp T,x\tp T}$ and
  $\seq P {\Gamma,\vec x\tp T}$ and $\tsf{subject}(\spref) = x$.
  By induction, we get $\seq{\spref\subs bx}{\Gamma, b\tp T}$ and
  $\seq{P\subs bx}{\Gamma,b\tp T}$.
  Applying the \tsf{(Weak)} rule as often as needed, we obtain
  $\seq{\spref\subs bx} {\Gamma, \vec b\tp T,b\tp T}$ and
  $\seq{P\subs bx}{\Gamma,\vec b\tp T}$.
  We can easily see that $\tsf{subject}(\spref\subs bx) = b$.
  We are then in a position to apply rule \tsf{(Catch)} to obtain
  $\seq{\trycatch {\spref\subs bx} P\subs bx}{\Gamma, \vec b\tp T}$,
  from which the result follows based the definition of substitution.
 
  \Case $a \neq x$.
  We have $\seq{\trycatch\spref P}{\Gamma, a\tp T,\vec x \tp U}$.
  By rule inversion we obtain $\seq \spref {\Gamma,a\tp T,\vec x\tp U}$ and
  $\seq P {\Gamma,\vec x\tp U}$ and $\tsf{subject}(\spref) = a$.
  Applying the induction hypothesis to both sequents, followed by the
  \tsf{(Catch)} rule we get
  $\seq{\trycatch {\spref\subs bx}{P\subs bx}}{\Gamma, a\tp T, b\tp
    U}$.
  The result follows by applying the \tsf{(Weak)} rule as often as
  needed and the definition of substitution.
\end{proof}

\begin{lem}[Small Subject Reduction]\
  \label{lem:ssr}
    If
    $\seq{\newc {ab} ( \aact_1 \prl \aact_2 )}{\Gamma}$ and
    $\newc {ab} ( \aact_1 \prl \aact_2 ) \R \newc {ab}\, Q$ 
    then
    $\seq{\newc {ab}\, Q}{\Gamma}$.
\end{lem}
\begin{proof}
  From $\seq{\newc {ab} (\aact_1 \prl \aact_2)}{\Gamma}$, using
  \tsf{(Contraction)} and \tsf{(Weak)} (zero or more times) and
  \tsf{(Res)}, we obtain
  $\seq{ \alpha_1 }{ \Gamma_1,\rT\Gamma_1',\rT\Gamma_1', a\tp T }$ and
  $\seq{ \alpha_2 }{ \Gamma_2,\rT\Gamma_2',\rT\Gamma_2', b\tp \dual
    T}$, where $\Gamma$ is
  $\Gamma_1,\rT\Gamma_1',\Gamma_2,\rT\Gamma_2', \rT \Theta, \sessend\,
  \Lambda$.
  We proceed by case
  analysis on the various reduction axioms, illustrating the most
  complex cases. 
  Notice that 
  $\tsf{subject}(\aact_1) = a$  
  and $\tsf{subject}(\aact_2) = b$, otherwise there is no reduction. 

  \Case \tsf{(R{-}Com)}. In this case $T$ is $\send T_1 . T_2$ and
  $\Gamma_1,\rT\Gamma_1',\rT\Gamma_1'$ is $\Gamma''_1, c\tp T_1$. We
  have four cases to consider depending on the shape of $H_1$ and
  $H_2$. When both $H_1$ and $H_2$ are $\tsf{do}$-$\tsf{catch}$
  contexts, using rules \tsf{(Contraction)}, \tsf{(Res)}, \tsf{(Catch)},
  \tsf{(Out)} and \tsf{(In)}, we continue with the only derivation scheme
  for the hypothesis, to conclude that
  $\seq{ P }{ \Gamma_1'', a \tp{T_2} }$ and
  $\seq{ Q }{ \Gamma_2,\rT \Gamma_2', \rT \Gamma_2', b \tp \dual{T_2}, x \tp T_1 }$.
  The result follows from the substitution lemma, and rules
  \tsf{(Contraction)}, \tsf{(Weak)} and \tsf{(Res)}.

  \begin{sloppypar}
    \Case \tsf{(C{-}Catch)}. Using rules \tsf{(Contraction)}, \tsf{(Res)},
    \tsf{(Cancel)}, and \tsf{(Catch)}, we conclude:
    $\seq{ \rho }{ \Gamma_1, \rT\Gamma'_1, \rT\Gamma_1', a \tp T }$
    and $\seq{ P }{ \Gamma'_1, \rT\Gamma''_1, \rT\Gamma''_1 }$ 
    where $\Gamma_2$ is the empty context.  
    We then distinguish two cases. When $a$ is free in $P$,
    from Lemma~\ref{lem:fv} we know that $a\tp U$ is in $\Gamma$,
    hence, by context formation, $U$ is a request type. We conclude
    the proof using rules \tsf{(Weak)}, \tsf{(Contraction)}, \tsf{(Cancel)}
    and \tsf{(Res)}. When $a$ is not free in $P$, we conclude the proof
    from $\seq{ P }{ \Gamma'_1, \rT\Gamma''_1, \rT\Gamma_1'' }$, using
    contraction and the derived rule \tsf{(Mix)}, recalling that
    $P \equiv \newc{a'b'} ( P \prl \nil )$ where $a',b'$ are fresh names. 
    (We first obtain $\newc{ab} ( P \prl \kils b) \equiv P$ and then introduce the 
    new binders which are $\alpha$-convertible to the desired result.)
  \end{sloppypar}

  \Case \tsf{(R{-}Ses)}.  We obtain
  $\seq{\aact_1}{\Gamma_1,\rT\Gamma_1',\rT\Gamma_1',a\tp \aT T}$ and
  $\seq{\aact_2}{\Gamma_2,\rT\Gamma_2',\rT\Gamma_2',b\tp \rT T}$.  The
  part $R$ of the scope in the rule is taken to be
  $\nil$.
  We consider two cases, depending on the shape of the $H$-context; we
  analyse the case where $H$ is empty since the exception handler
  disappears anyway.  Rules \tsf{(Contraction)}, \tsf{(Acc)},
  \tsf{(Req)} allow to conclude that
  $\seq{P}{\Gamma_1,\rT\Gamma_1',\rT\Gamma_1'}$ and
  $\seq{Q}{\Gamma_2,\rT\Gamma_2',\rT\Gamma_2', x\tp U}$. The result
  follows from the application of the substitution lemma, and rules
  \tsf{(Contraction)}, \tsf{(Weak)}, \tsf{(Mix)}, and \tsf{(Res)} if
  $a$ appears in~$P$.
\end{proof}

\subsection*{Theorem~\ref{thm:sr} (Subject Reduction)}
If $\seq P \Gamma$ and $P \R Q$ then $\seq Q \Gamma$.
\begin{proof}
  We proceed by induction on the typing derivation.  Since $P$ must
  contain a redex, the possible last rules for any derivation of
  $\seq P \Gamma$ are \tsf{(Res, Contraction, Weak)}.  From these, the
  cases for \tsf{(Contraction, Weak)} follow immediately by the
  induction hypothesis.  Thus, we can focus on \tsf{(Res)}.

\Case \tsf{(Res)}
We have $\tyderT{}{\seq P \Gamma}$ and the conclusion is of the shape
$\seq{ \newc{ab} ( P_1 \prl P_2 ) }{ \Gamma_1, \Gamma_2 }$ with
premises $\tyderT{{11}}{\seq{P_1}{\Gamma_1, a \tp T}}$ and
$\tyderT{{21}}{\seq{P_2}{\Gamma_2, b \tp \dual T}}$.  (We indicate the
derivations because we need to manipulate them.)

Now we consider the possible reductions of $P$.  If $P_i \R P_i'$ then
the result easily follows from the induction hypothesis followed by an
application of \tsf{(Res)}.  We therefore focus on the case where
(part of) $P_1$ interacts with (part of) $P_2$.  In this case there is
only one possibility, that of a reduction on the endpoints $a$ and $b$
(possibly more than one when $T$ or $\dual T$ is of the shape
$\rT {T'}$), which follows by the well-formedness of
$\Gamma_1, \Gamma_2$.  (Recall that the $\Gamma_i$ can only have
common elements of the shape $c \tp \rT T'$, so it follows that the
$P_i$ can only communicate over the interface provided by $a \tp T$
and $b \tp \dual T$.)

We now consider the possible sub-cases for a redex on $a/b$. 
In order for a redex to be active, both components obviously  
need to be at the top level (i.e., unguarded), 
so we can determine that:%
\begin{gather*}
P_1 \equiv \newc{\overrightarrow{cd}}( \aact_1 \prl P_1' ) 
\qquad
P_2 \equiv \newc{\overrightarrow{rs}}( \aact_2 \prl P_2' )
\end{gather*}
where the $\aact_i$ are dual actions on $a$ and $b$, respectively.  

We thus have $P \R \newc{ab}\, Q$ where $Q$ is
\begin{equation*}
  \newc{\overrightarrow{cd}} \newc{\overrightarrow{rs}} ( Q_1 \prl P_1' \prl P_2')
\end{equation*}
up to structural equivalence as usual.

Now, using Lemma~\ref{lem:atomic-typa} we obtain: %
\begin{equation*}
\tyderC{11}
           { \tyderT{12}{\seq{\aact_1}{\Delta_1, a \tp T}} }
           {\seq{P_1}{\Gamma_1, a \tp T}}
\qquad
\tyderC{21}
           { \tyderT{22}{\seq{\aact_2}{\Delta_2, b \tp \dual T}} }
           {\seq{P_2}{\Gamma_2, b \tp \dual T}}
\end{equation*}%
Moreover, by Proposition~\ref{prop:iorder} we obtain that 
$\Delta_1, \Delta_2$ is a well-formed interface. 

We now consider two cases, where the redices are linear and unrestricted, 
respectively. 
This is because each case requires a slightly different construction 
in order to obtain a derivation for the contractum. 
We start with the linear case, in which we know that 
the $\aact_i$ are the only actions with subject $a$ and $b$  
(for the other case we need to take into account the 
fact that multiple actions with interface  $b \tp \rT T'$ can appear).

\paragraph{\textbf{Linear redex}}
We form the following derivation:%
\begin{displaymath}
\tyder_3 \qquad = \qquad 
\dfrac{\tyderT{12}{\seq{\aact_1}{\Delta_1, a \tp T}}
          \qquad 
          \tyderT{22}{\seq{\aact_2}{\Delta_2, b \tp \dual T}} }
        {\seq{\newc{ab}( \aact_1 \prl \aact_2 )}{\Delta_1, \Delta_2}}
        \quad \mathsf{(Res)}
\end{displaymath}%
We have $\newc{ab}( \aact_1 \prl \aact_2 ) \R \newc {ab} Q_1$ and 
by Lemma~\ref{lem:ssr} and $\tyder_3$ we obtain 
\tyderT{4}{\seq{\newc {ab} Q_1}{\Delta_1, \Delta_2}}. 
Now we can obtain the following derivation:%
\begin{displaymath}
               \tyderC{11}{ \tyderC{21}{ \tyder_4 }{\Gamma_2,\Delta_1} } 
               {\seq{Q}{\Gamma_1, \Gamma_2}}
\end{displaymath}%
It is easy to check that the conditions of Definition~\ref{def:dercomp} 
are respected (see also Proposition~\ref{prop:subder}), i.e., that 
the result is a valid derivation as required.

\paragraph{\textbf{Unrestricted redex}}
Let us assume, without loss of generality (the other case is symmetric), 
that $T = \rT T'$.  
In this case we must take into account that multiple compositions against 
$b \tp \aT \dual{T'}$ are possible within the derivation, so we perform 
a few more steps. 

 \Case \tsf{(R{-}Ses)}.  
We have the following derivations, assuming 
$\aact_1 = \req a c . R_1$ 
and 
$\aact_2 = \acc b x. R_2$, 
which are obtained by inversion 
for each prefix, 
taking into account weakening and contraction:
\begin{gather*}
\mathcal{D}_{12} = 
\dfrac{  
	\dfrac{ \tyderT{13}{\seq{R_1}{\Delta_1', \rT \Sigma_1, \rT \Sigma_1}}  }
			 {\seq{\req a c . R_1}{\Delta_1', \rT \Sigma_1, \rT \Sigma_1, a \tp \rT {T'}, c \tp T' }} \quad \tsf{(Req)}
}{ 
    \overline{ 
    	\:\: \seq{\req a c . R_1}{\Delta_1', \rT \Sigma_1, \rT \Theta_1, \sessend\, \Lambda_1, a \tp \rT {T'}, c \tp T'}  \:\: 
    }
} \quad \tsf{(Contraction), (Weak)}
\\[.5ex]
\mathcal{D}_{22} = 
\dfrac{  
	\dfrac{ \tyderT{23}{\seq{R_2}{\Delta_2', \rT \Sigma_2, \rT \Sigma_2, x \tp T'} }  }
			 {\seq{\acc b x . R_2}{\Delta_2', \rT \Sigma_2, \rT \Sigma_2, b \tp \aT {T'}} } \quad \tsf{(Acc)}
}{ 
    \overline{ 
    	\:\: \seq{\acc b x . R_2}{\Delta_2', \rT \Sigma_2, \rT \Theta_2, \sessend\, \Lambda_2, b \tp \aT {T'}} \:\: 
    }
} \quad \tsf{(Contraction), (Weak)}
\end{gather*}%
with $\Delta_1 = \Delta_1', \rT \Sigma_1, \rT \Theta_1, \sessend\, \Lambda_1, c \tp T'$  
and $\Delta_2 = \Delta_2', \rT \Sigma_2, \rT \Theta_2, \sessend\, \Lambda_2$ 
with 
$\Delta_2' = x_1 \tp \rT {T_1}, \ldots, x_n \tp \rT {T_n}$. 

From the substitution lemma we obtain
$\tyderT{24}{\seq{R_2 \subs c x}{(\Delta_2', \rT \Sigma_2, \rT
  \Sigma_2, x \tp T')\subs c x }}$.

First we form:
\begin{displaymath}
\tyder_{14} \qquad = \qquad 
\dfrac{\tyder_{13} 
          \qquad 
          \tyder_{24} }
        {\seq{R_1 \prl R_2 \subs c x}{\Delta_1', \rT \Sigma_1, \rT \Sigma_1, (\Delta_2', \rT \Sigma_2, \rT \Sigma_2, x \tp T')\subs c x}} \quad \mathsf{(Mix)}
\end{displaymath}%
The above interface is clearly well-formed. First, by Proposition~\ref{prop:iorder}  
$\Delta_1', \Delta_2'$ is well-formed, and also $c \tp T'$ is composable with $\Delta_1'$ by assumption. 

Now, we form the following derivation, which is equal to the original except that we put 
$\tyder_{14}$ in the place of $\tyder_{12}$:%
\begin{displaymath}
\tyder_3 \qquad = \qquad 
\dfrac{ \tyderCs{11}{ \tyder_{14}  } 
          \qquad 
          \tyder_{21} }
        {\seq{\newc{ab} Q}{\Gamma_1, \Gamma_2]}} \quad \mathsf{(Res)}
\end{displaymath}%
We omit some applications of \tsf{(Contraction)} on $\tyderCs{11}{ \tyder_{14}  }$, 
as well as below \tsf{(Res)},  
which are needed to obtain the original interface, as required. 

 \Case \tsf{(C{-}Acc/Req/Cat)}. Similar to the above.  
\end{proof}


\end{document}